%% file: fossacs2011.tex
\documentclass{article}
\usepackage[english]{babel}
\usepackage[utf8]{inputenc}

\usepackage{amsthm}
\usepackage{amssymb,stmaryrd}
\usepackage[all]{xy}
\usepackage{url}
\usepackage{graphicx}
\usepackage{turnstile}

\newcommand{\restrict}{\upharpoonright}
\newcommand{\from}{\leftarrow}
\newcommand{\tto}{\Rightarrow}

\newcommand{\intr}[1]{\llbracket #1 \rrbracket}

\newcommand{\inter}{\mathop{|\!|}}

\newcommand{\unfold}[1]{\widetilde{A}}

\newcommand{\implies}{\Rightarrow}
\newcommand{\vis}{\mathcal{V}}
\newcommand{\church}[1]{\underline{#1}}
\newcommand{\size}{\mathrm{rsize}}
\newcommand{\cosize}{\mathrm{rcosize}}

\newtheorem{lemma}{\textsc{Lemma}}
\newtheorem{theorem}{\textsc{Theorem}}
\newtheorem{definition}{\textsc{Definition}}
\newtheorem{proposition}{\textsc{Proposition}}

\newlength{\viewht}
\newlength{\viewlift}
\newlength{\viewdp}
\newlength{\viewdrop}

\newcommand{\pview}[1]{
\settoheight{\viewht}{\makebox{$#1$}}
\setlength{\viewlift}{\viewht}%
\addtolength{\viewlift}{-1ex}%
\raisebox{0.3\viewlift}{
  \makebox{$\ulcorner$}}
  \!#1\!
\settoheight{\viewht}{\makebox{$#1$}}
\setlength{\viewlift}{\viewht}%
\addtolength{\viewlift}{-1ex}%
\raisebox{0.3\viewlift}{
  \makebox{$\urcorner$}}
}

\newcommand{\oview}[1]{
\settodepth{\viewdp}{\makebox{$#1$}}
\setlength{\viewdrop}{0.3\viewdp}%
\addtolength{\viewdrop}{0.5ex}%
\raisebox{-\viewdrop}{
  \makebox{$\llcorner$}}
  \!#1\!
\settodepth{\viewdp}{\makebox{$#1$}}
\setlength{\viewdrop}{0.3\viewdp}%
\addtolength{\viewdrop}{0.5ex}%
\raisebox{-\viewdrop}{
  \makebox{$\lrcorner$}}
} 

\begin{document}
\title{Estimation of the length of interactions\\ in arena game semantics}
\author{Pierre Clairambault\\
University of Bath\\
\texttt{p.clairambault@bath.ac.uk}}
\date{}
\maketitle
\begin{abstract}
We estimate the maximal length of interactions between strategies in HO/N game semantics, in the spirit of the work by Schwichtenberg and Beckmann for the length of reduction in 
simply typed $\lambda$-calculus.
Because of the operational content of game semantics, the bounds presented here also apply to head linear reduction on $\lambda$-terms and to the execution of
programs by abstract machines (PAM/KAM), including in presence of computational effects such as non-determinism or ground type references. The proof proceeds by extracting from
the games model a combinatorial rewriting rule on trees of natural numbers, which can then be analysed independently of game semantics or $\lambda$-calculus.
\end{abstract}

\section{Introduction}

Among the numerous notions of execution that one can consider on higher-order programming languages (in particular on the $\lambda$-calculus) \emph{head linear reduction} \cite{danos:abstract} plays a particular role. Although
it is not as widespread and specifically studied as, say, $\beta$-reduction, it is nonetheless implicit to various approaches of higher-order computation, such as geometry of interaction, game semantics, optimal
reduction and ordinary operational semantics. It is also implicit to several abstract machines, including the Krivine Abstract Machine (KAM) \cite{krivine1985interpreteur} and the Pointer Abstract Machine (PAM) \cite{danos:abstract}, 
in the sense that it is the reduction they perform \cite{danos:abstract,phd} and as such is a valuable abstraction of how programs are executed in the implementation of higher order languages. 

Despite being closer to the implementation of programming languages, head linear reduction never drew a lot of attention from the community. Part of the reason for that is that it is not a usual notion of reduction:
defining it properly on $\lambda$-terms requires both to extend the notion of redex and to restrict to linear substitution, leading to rather subtle and tricky definitions which lack the canonicity of 
$\beta$-reduction\footnote{However, there are syntaxes on which head linear reductions appear more canonical that $\beta$-reduction, for instance proof nets \cite{DBLP:journals/tcs/MascariP94}.}. 
Moreover, its associated observational equivalence is the same as for the usual head $\beta$-reduction, which makes it non relevant as long as one is
interested in the equational theory of $\lambda$-calculus.
However, head linear reduction should appear in the foreground as soon as one is interested in quantitative aspects of computation, such as complexity. On the contrary, although very precise bounds are known for the possible length of $\beta$-reduction chains in simply typed $\lambda$-calculus \cite{schwichtenberg1982complexity,beckmann2001exact}, to the author's knowledge, the situation for head linear reduction remains essentially
unexplored. Even if it is generally expected that the bounds remain hyper-exponential\footnote{Note however that in \cite{de1987generalizing}, De Bruijn gives an upper bound for his local $\beta$-reduction, akin to head linear
reduction. The bound is an iterate of the diagonal of an Ackermann-like function!}  (and this indeed what we will prove), it does not seem to follow easily from the bounds known for $\beta$-reduction.

Rather than reasoning directly on head linear reduction, we will instead look at it through game semantics \cite{hyland-ong}. Indeed, there is a close relationship between head linear reduction and interaction in games
model of programming languages \cite{danosregnier}. More precisely, given two $\beta$-normal and $\eta$-long $\lambda$-terms $S$ and $T$, there is a step-by-step correspondence between head linear reduction chains of $S T$ and 
game-theoretic interactions between the strategies $\intr{S}$ and $\intr{T}$. Of course, game semantics are not central to our analysis: as is often the case, our methods and results could be adapted to a purely syntactical
framework. However, games have this considerable advantage of accommodating in a single framework purely functional programming languages such as the $\lambda$-calculus or PCF and a number of computational features such as 
non-determinism \cite{DBLP:conf/lics/HarmerM99}, control operators \cite{DBLP:conf/lics/Laird97} and references \cite{abramsky-mccusker:active-algol}. This will allow us to do our study with an increased generality:
our complexity results will hold for a variety of settings, from simply typed $\lambda$-calculus to richer languages possibly featuring the computational effects mentioned above, as long as there is no
fixed point operator.

\paragraph{Outline.} In Section 2 we will recall some of the basic definitions of Hyland-Ong game semantics, define the central notion of size of a strategy, and introduce our main question as the problem of finding 
the maximal length of an interaction between two strategies of fixed size. Our approach will be then to progressively simplify this problem in order to reach its underlying combinatorial nature. In Section 3
we first introduce the notion of \emph{visible pointer structures}, \emph{i.e.} plays where the identity of moves has been forgotten. This allows a more elementary (strategy-free) equivalent statement of our problem.
Then we show how each position in a visible pointer structure can be characterised by a tree of natural numbers called an \emph{agent}. We then show that the problem can be once again reformulated as the maximal length of
a reduction on these agents. In Section 4 we study the length of this reduction, giving in particular an upper bound. We also give a corresponding lower bound, and finally use our result to estimate
the maximal length of head linear reduction sequences on simply typed $\lambda$-terms.

\paragraph{Related works.} Our results and part of our methods are similar to the works of Schwichtenberg and Beckmann \cite{schwichtenberg1982complexity,beckmann2001exact}, but the reduction we study is in some sense more
challenging, because redexes are not
destroyed as they are reduced. Moreover, the game semantics setting allows for an extra generality. The present work also has common points with work by Dal Lago and Laurent \cite{DBLP:conf/csl/LagoL08}, in the sense that it uses
tools from game semantics to reason on the length of execution. However the approach is very different : their estimate is very precise but uses an information on terms difficult to compute (almost as hard as actually performing
execution). Here, we need little information on terms (gathering this information is linear in the size of the term), but our bounds are, in most cases, very rough.

\section{Arena game semantics}
\label{section_games}

We recall briefly the now usual definitions of arena games, first introduced in \cite{hyland-ong}. More detailed accounts
can be found in \cite{fpc2000,harmer2004innocent}. We are interested in games with two participants: Opponent
(O, the \emph{environment}) and Player (P, the \emph{program}).

\subsection{Arenas and Plays}

Valid plays are generated by directed graphs called \emph{arenas}, which are semantic versions of \emph{types}. 
Formally, an \textbf{arena} is a structure $A=(M_A,\lambda_A,I_A,\vdash_{A})$ where:
\begin{itemize}
\item $M_{A}$ is a set of \textbf{moves},
\item $\lambda_{A}:M_{A} \to \{O,P\}$ is a polarity function indicating whether
a move is an Opponent or Player move ($O$-move or $P$-move).
\item $I_A\subseteq \lambda_A^{-1}(\{O\})$ is a set of \textbf{initial moves}.
\item $\vdash_{A} \subset M_{A}\times M_A$ is a relation called \textbf{enabling}, such that
if $m \vdash_A n$, then $\lambda_{A}(m)\neq \lambda_{A}(n)$.
\end{itemize}
In other words, an arena is just a directed bipartite graph. We now define plays as \textbf{justified sequences} over ${A}$: these are
sequences $s$ of moves of ${A}$, each non-initial move $m$ in $s$ being equipped with a pointer to an earlier move
$n$ in $s$, satisfying $n\vdash_{A} m$. In other words, a justified sequence $s$ over ${A}$ is such that
each reversed pointer chain $s_{i_0}\from s_{i_1} \from \dots \from s_{i_n}$ is a path on ${A}$ (viewed as a graph).
The role of pointers is to allow \emph{reopenings} or \emph{backtracking} in plays. When writing justified sequences, we will often omit the
justification information if this does not cause any ambiguity. The symbol $\sqsubseteq$ will denote the prefix ordering on justified
sequences, and $s_1 \sqsubseteq^P s_2$ will mean that $s_1$ is a $P$-ending prefix of $s_2$.
If $s$ is a justified sequence on ${A}$, $|s|$ will denote its length.

Given a justified sequence $s$ on ${A}$, it has two subsequences of particular interest: the P-view and O-view.
The view for P (resp. O) may be understood as the subsequence of the play where P (resp. O) only sees his own duplications.
Practically, the P-view $\pview{s}$ of $s$ is computed by forgetting everything
under Opponent's pointers, in the following recursive way:
\begin{itemize}
\item $\pview{sm} = \pview{s}m$ if $\lambda_{A}(m)=P$;
\item $\pview{sm} = m$ if $m\in I_A$ and $m$ has no justification pointer;
\item $\pview{s_1m s_2n} = \pview{s}mn$ if $\lambda_{A}(n)=O$ and $n$ points to $m$.
\end{itemize}
The O-view $\oview{s}$ of $s$ is defined dually, without the special treatment of initial moves\footnote{In the terminology of \cite{harmer2004innocent}, it is the \emph{long} $O$-view.}. 
The \emph{legal plays} over ${A}$, denoted by $\mathcal{L}_{A}$, are the
justified sequences $s$ on ${A}$ satisfying the \textbf{alternation} condition, \emph{i.e.} that 
if $tmn \sqsubseteq s$, then $\lambda_{A}(m)\neq \lambda_A(n)$.

\subsection{Classes of strategies}

In this subsection, we will present several classes of strategies on arena games that are of interest to us in the present paper. 
A \textbf{strategy} $\sigma$ on ${A}$ is a set of even-length legal plays on ${A}$, closed under even-length prefix. A strategy from $A$ to $B$
is a strategy $\sigma: A\tto B$, where $A\tto B$ is the usual arrow arena defined by $M_{A\tto B} = M_A + M_B$, $\lambda_{A\tto B} = [\overline{\lambda_A}, \lambda_B]$ (where
$\overline{\lambda_A}$ means $\lambda_A$ with polarity $O/P$ reversed), $I_{A\tto B} = I_B$ and $\vdash_{A\tto B} = \vdash_A + \vdash_B + I_B\times I_A$.

\paragraph{Composition.} We define composition of strategies by the usual parallel interaction plus hiding mechanism.
If ${A}$, ${B}$ and ${C}$ are arenas, we define the set of \textbf{interactions}
$I({A},{B},{C})$ as the set of justified sequences $u$ over ${A}$, ${B}$
and ${C}$ such that $u_{\restrict_{{A},{B}}}\in \mathcal{L}_{{A}\tto {B}}$,
$u_{\restrict_{{B},{C}}}\in \mathcal{L}_{{B}\tto {C}}$ and
$u_{\restrict_{{A},{C}}}\in \mathcal{L}_{{A}\tto{C}}$. Then, if $\sigma:{A}\tto {B}$
and $\tau:{B}\tto {C}$, we define parallel interaction as $\sigma \inter \tau = \{u\in I({A},{B},{C}) \mid   u_{\restrict A, B}\in \sigma \wedge u_{\restrict B, C} \in \tau\}$,
composition is then defined as $\sigma;\tau = \{u_{\restrict A, C} \mid  u\in \sigma||\tau\}$. Composition is associative and admits copycat strategies 
as identities.

\paragraph{$P$-visible strategies.} A strategy $\sigma$ is \textbf{$P$-visible} if each of its moves points to the current $P$-view. Formally, for all $sab \in \sigma$,
$b$ points inside $\pview{sa}$. $P$-visible strategies are stable under composition, as is proved for instance in \cite{harmer2004innocent}. They
correspond loosely to functional programs with ground type references \cite{abramsky-mccusker:active-algol}.

\paragraph{Innocent strategies.}
The class of \emph{innocent} strategies is central in game semantics, because of their correspondence with purely functional programs (or 
$\lambda$-terms) and of their useful definability properties. A strategy $\sigma$ is \textbf{innocent} if 
\[
sab\in \sigma \wedge t\in \sigma \wedge ta\in \mathcal{L}_{A} \wedge \pview{sa}=\pview{ta} \implies tab\in \sigma
\]
Intuitively, an innocent strategy only takes its $P$-view into account to determine its next move. Indeed, any innocent strategy is characterized by
a set of $P$-views. This observation is very important since $P$-views can be seen as abstract representations of branches of $\eta$-expanded Böhm trees
(\emph{a.k.a.} Nakajima trees \cite{nakajima1975infinite}) : this is the key to the definability process on innocent strategies \cite{hyland-ong}. It is quite technical to prove that
innocent strategies are stable under composition, proofs can be found for instance in \cite{harmer2004innocent,phd}. Arenas and innocent strategies form a cartesian closed category and
are therefore a model of simply typed $\lambda$-calculus.

\paragraph{Bounded strategies.} A strategy $\sigma$ is \textbf{bounded} if it is $P$-visible and if the length of its $P$-views is bounded: formally, there exists
$N\in \mathbb{N}$ such that for all $s\in \sigma$, $|\pview{s}|\leq N$. Bounded strategies are stable under composition, as is proved in \cite{totality} for the innocent case
and in \cite{phd} for the general case. This result corresponds loosely to the strong normalisation result on simply-typed $\lambda$-calculus.
Syntactically, bounded strategies include the interpretation of all terms of a functional
programming language without a fixed point operator but with \textsc{Algol}-like ground type references (for details about how reference cells get
interpreted as strategies see for instance \cite{abramsky-mccusker:active-algol}, it is obvious that this interpretation yields a bounded strategy) and arbitrary non determinism. This remark
is important since it implies that our results will hold for any program written with these constructs, as long as they do not use recursion or a fixed point
operator. 

\subsection{Size of strategies and interactions}

Since in this paper we will be interested in the length of interactions, it is sensible to make it precise first what we mean by the \emph{size} of strategies. 
Let $\sigma$ be a bounded strategy, its size is defined as
\[
|\sigma| = \frac{max_{s\in \sigma} |\pview{s}|}{2}
\]
All our analysis on the size of interactions will be based on this notion of size of strategies. Our starting point is the following finiteness result, 
proved in \cite{totality}. We say that an interaction $u\in I(A, B, C)$ is \textbf{passive} if the only move by the external Opponent on $A, C$
is the initial move on $C$, so that the interaction stops as soon as we need additional input from the external Opponent. 

\begin{proposition}
Let $\sigma : A\tto B$ and $\tau: B \tto C$ be bounded strategies and let $u\in \sigma || \tau$ be a passive interaction, then $u$ is finite.
\label{finiteness}
\end{proposition}

Using this, we can actually deduce the existence of an uniform bound on the length of such $u\in \sigma||\tau$, which only depends
on the respective size of $\sigma$ and $\tau$:

\begin{lemma}
For all $n, p\in \mathbb{N}$ there is a lesser $N(n, p)\in \mathbb{N}$ such that for all arenas $A, B$ and $C$, for all $\sigma: A\tto B$ and $\tau: B\tto C$ such that
$|\sigma|\leq p$ and $|\tau|\leq n$, for all passive $u\in \sigma||\tau$ we have $|u|\leq N(n,p)$.
\label{konig}
\end{lemma}
\begin{proof}
For arenas $A, B$ and $C$ consider the set $T_{A, B, C}$ of all passive interactions $u\in I(A, B, C)$ such that for all $s\sqsubseteq u_{\restrict B, C}$,
$|\pview{s}|\leq 2n$ and for all $s\sqsubseteq u_{A, B}$, $|\pview{s}|\leq 2p$. Then, consider the union $T$ of all the $T_{A, B, C}$, our goal
here is to find a bound on the length of all elements of $T$. Consider now the tree structure on $T$ given by the prefix ordering.
To make this tree finitely branching, consider the relation $m \cong n \Leftrightarrow depth(m) = depth(n)$ on moves, where $depth(m)$ is the number of pointers required
to go from $m$ to an initial move. The tree $T/\cong$ is now finitely branching, but is also well-founded by Proposition \ref{finiteness}, therefore it is finite by
König's lemma\footnote{Or, more adequately, the fan theorem.}. Let $N(n,p)$ be its maximal depth, it is now obvious that it satisfies the required properties.
\end{proof}

We have proved the existence of the uniform bound $N(n,p)$, but in a way that provides no feasible means of estimating $N(n,p)$. The goal of the rest of this
paper is to estimate this bound as precisely as possible. As a matter of fact, we will be mainly interested in the ``typed" variant $N_d(n,p)$, defined as the maximum
length of all possible passive interactions between strategies $\sigma: A\tto B$ and $\tau: B \tto C$ of respective size $p$ and $n$, where $B$ has a finite depth $d-1$.

\section{Pointer structures and rewriting}
We have seen that to prove Lemma \ref{konig}, we must consider plays up to an equivalence relation $\cong$ which assimilates all moves at the same depth. Indeed,
general arenas and plays contain information which is useless for us. Following \cite{totality}, we will here reason on
\emph{pointer structures}, which result of considering moves in plays up to $\cong$.  Pointer structures are also similar to the \emph{parity pointer functions} of
Harmer, Hyland and Melliès \cite{hhm} and to the \emph{interaction sequences} of Coquand \cite{coquand}. We will delve here into their combinatorics and extract from them a small rewriting system,
whose study is sufficient to characterize their length.

\subsection{$N_d(n,p)$ as a bound for pointer structures}

\paragraph{Visible pointer structures.}
In \cite{totality}, we introduced pointer structures by elementary axioms, independent of the general notions of game semantics. Instead here, we define \textbf{pointer
structures} as usual alternating plays, but on the particular ``pure" arena $I_\omega = \bigsqcup_{n\in \mathbb{N}} I_n$, where $I_0 = \bot$ ($\bot$ is the singleton
arena with just one Opponent move) and $I_{n+1} = I_n \tto \bot$. 
As we are interested in the interaction between $P$-visible strategies, we will only consider \textbf{visible} pointer structures,
where both players point in their corresponding view. Formally, $s$ is visible if for all $s' p\sqsubseteq^P s$, 
$a$ points inside $\pview{s'}$ and if for all $s' o \sqsubseteq^O s$, $o$ points inside $\oview{s'}$. The \textbf{depth} of a visible pointer structure $s$ is the smallest $d$
such that $s$ is a play on $I_d$. Let us denote by $\vis$ the set of all visible pointer structures.

\paragraph{Atomic agents.}
After forgetting information on plays, let us forget information on strategies. Instead of considering bounded strategies with all their intentional behaviour,
we will just keep the data of their size. Pointer structures will then be considered as interactions between the corresponding numbers which will be called
\textbf{atomic agents}. If $n$ is such a natural number, we define its \textbf{trace} as follows, along with the dual notion of \textbf{co-trace}:
\begin{eqnarray*}
Tr(n) &=& \{s \in \vis \mid \forall s'\sqsubseteq s, |\pview{s'}|\leq 2n\}\\
coTr(p) &=& \{s\in \vis \mid \forall s'\sqsubseteq s, |\oview{s'}| \leq 2p+1\}
\end{eqnarray*}
An \textbf{interaction} at depth $d$ between $n$ and $p$ is a visible pointer structure $s$ of depth at most $d$ such that $s\in Tr(n)\cap coTr(p)$.
We write $s\in n \star_d p$. These definitions allow to give the following strategy-free equivalent formulation of $N_d(n,p)$.

\begin{lemma}
Let $n$ and $p$ be natural numbers and $d\geq 2$, then
\[
N_d(n,p) = max \{ |s| \mid s\in n \star_d p\}
\]
\end{lemma}
\begin{proof}
Consider the maximal bounded strategies of respective size $n$ and $p$, defined as $\mathbf{n} = \{s\in I_{d}\mid \forall s'\sqsubseteq s,~|\pview{s'}|\leq 2n$ 
and $\mathbf{p} = \{s\in I_{d-1} \mid \forall s'\sqsubseteq s,~|\pview{s'}|\leq 2p$. Then pointer structures in $n \star_d p$ are the same as (passive) interactions
in $\mathbf{p}||\mathbf{n}$, thus $max \{ |s| \mid s\in n \star_d p\} \leq N_d(n,p)$. Reciprocally, if $\sigma: A\tto B$ has size $p$ and $\tau: B\tto C$ has size
$n$ and if $u\in \sigma||\tau$ is passive, then if $u'$ denotes $u$ where moves are considered up to $\cong$ we have $u' \in \mathbf{p}||\mathbf{n}$
thus $u' \in n \star_d p$ and $N_d(n,p) = max \{ |s| \mid s\in n \star_d p\}$.
\end{proof}

\subsection{Agents}

To bound the length of a pointer structure $s$, our idea is to label each of its moves $s_i$ by an object $t$, expressing the size that the strategies have left. Let us consider
here an analogy between pointer structures and the execution of $\lambda$-terms by the KAM\footnote{The syntax used here seems natural enough, but is for instance described in \cite{phd}.}.
Consider the following three KAM computation steps:
\begin{eqnarray*}
(\lambda x.x S) \star T \cdot \pi_0     &\leadsto^3& T \star S^{x\mapsto T} \cdot \pi_0
\end{eqnarray*}
The interaction between two closed terms (with empty environment) leads, after three steps of computation, to the interaction between two \emph{open terms}
$T$ and $S$ (where $x$ is free in $S$), with an environment. By analogy, if $s_0$ is labelled by the pair $(n,p)$ of interacting ``strategies", each move $s_i$
should correspond to an interaction between objects $(a, b)$, where $a$ and $b$ have a tree-like structure which is reminiscent of those of closures\footnote{As in the example above, closures
are pairs $M^\sigma$ where $M$ is an open term and $\sigma$ is an \emph{environment}, \emph{i.e.} a mapping which to each free variable of $M$ associates a closure.}.

We will call a \textbf{pointed visible pointer structure} (pvps) a pair $(s, i)$ where $s$ is a visible pointer structure and $i\leq |s|-1$
is an arbitrary ``starting" move. We adapt the notions of size and depth for pvps, and introduce a notion of \emph{context}.

\begin{definition}
Let $(s,i)$ be a pointed visible pointer structure. The \textbf{residual size} of $s$ at $i$, written $\size(s, i)$, is defined as follows:
\begin{itemize}
\item If $s_i$ is an Opponent move, it is $\max_{s_i \in \pview{s_{\leq j}}} |\pview{s_{\leq j}}| - |\pview{s_{\leq i}}| + 1$
\item If $s_i$ is a Player move, it is $\max_{s_i \in \oview{s_{\leq j}}} |\oview{s_{\leq j}}| - |\oview{s_{\leq i}}| + 1$
\end{itemize}
where $s_i \in \pview{s_{\leq j}}$ means that the computation of $\pview{s_{\leq j}}$ reaches\footnote{So starting from $s_j$ and following Opponent's pointers
eventually reaches $s_i$.} $s_i$. Dually, we have the notion of
\textbf{residual co-size} of $s$ at $i$, written $\cosize(s, i)$, defined as follows:
\begin{itemize}
\item If $s_i$ is an Opponent move, it is $\max_{s_i \in \oview{s_{\leq j}}} |\oview{s_{\leq j}}| - |\oview{s_{\leq i}}| + 1$
\item Otherwise, $\max_{s_i \in \pview{s_{\leq j}}} |\pview{s_{\leq j}}| - |\pview{s_{\leq i}}| + 1$
\end{itemize}
The \emph{residual depth} of $s$ at $i$ is the maximal length of a pointer chain in $s$ starting from $s_i$.
\end{definition}

\begin{definition}
Let $s$ be a visible pointer structure. We define the \textbf{context} of $(s,i)$ as:
\begin{itemize}
\item If $s_i$ is an O-move, the set $\{s_{n_1}, \dots, s_{n_p}\}$ of O-moves appearing in $\pview{s_{< i}}$,
\item If $s_i$ is a P-move, the set $\{s_{n_1}, \dots, s_{n_p}\}$ of P-moves appearing in $\oview{s_{< i}}$.
\end{itemize}
In other words it is the set of moves to which $s_{i+1}$ can point whilst abiding to the visibility condition, except $s_i$. We also need the dual notion
of co-context, which contains the moves the other player can point to. The \textbf{co-context} of $(s, i)$ is:
\begin{itemize}
\item If $s_i$ is an O-move, the set $\{s_{n_1},\dots, s_{n_p}\}$ of P-moves appearing in $\oview{s_{< i}}$,
\item If $s_i$ is a P-move, the set $\{s_{n_1},\dots, s_{n_p}\}$ of O-moves appearing in $\pview{s_{< i}}$.
\end{itemize}
\end{definition}

\begin{definition}
A \textbf{general agent} (just called agent for short) is a finite tree, whose nodes and edges are both labelled by natural numbers. If $a_1, \dots, a_p$ are agents and $d_1, \dots, d_p$ are natural numbers, we write:
\[
n[\{d_1\} a_1, \dots, \{d_p\} a_p] = 
\raisebox{20pt}{\xymatrix{
&n	\ar@{-}[dl]_{d_1}
	\ar@{-}[dr]^{d_p}\\
a_1&\dots&a_p
}}
\]
\end{definition}

\begin{definition}[Trace, co-trace, interaction]
Let us generalize the notion of trace to general agents. The two notions $Tr$ and $coTr$ are defined by mutual recursion, as follows:
let $a = n[\{d_1\}a_1, \dots, \{d_p\}a_p]$ be an agent. We say that $(s,i)$ is a \textbf{trace} (resp. a \textbf{co-trace}) of $a$, denoted $(s,i)\in Tr(a)$ (resp.
$(s, i)\in coTr(a)$) if the following conditions are satisfied:
\begin{itemize}
\item $\size(s, i) \leq 2n$ (resp. $\cosize(s, i) \leq 2n+1$),
\item If $\{s_{n_1}, \dots, s_{n_p}\}$ is the context of $(s,i)$ (resp. co-context), then for each $k\in \{1, \dots, p\}$ we have $(s,n_k) \in coTr(a_k)$.
\item If $\{s_{n_1}, \dots, s_{n_p}\}$ is the context of $(s,i)$ (resp. co-context), then for each $k\in \{1, \dots, p\}$ the residual depth of $s$ at $n_k$ is less than $d_k$.
\end{itemize}
Then, we define an \textbf{interaction} of two agents $a$ and $b$ at depth $d$ as a pair $(s, i)\in Tr(a)\cap coTr(b)$ where the residual depth of $s$ at $i$ is less than $d$, which we write $(s, i)\in a\star_d b$.
\end{definition}
Notice that we use the same notations $Tr$, $coTr$ and $\star$ both for natural numbers and general agents. This should not generate any confusion, since the definitions
just above coincide with the previous ones in the special case of ``atomic", or closed, agents: if $n$ and $p$ are natural numbers, then obviously $s\in n \star_d p$ if and
only if $(s, 0)\in n[] \star_d p[]$.
Note also that definitions are adapted here to this particular setting where strategies are replaced by natural numbers, however they could be generalized to the usual notion
of strategies. An agent would be then a tree of strategies, and a trace of this agent would be a possible interaction between all these strategies. This would
be a new approach to the problem of \emph{revealed} or \emph{uncovered} game semantics \cite{greenland2004game,blum2008thesis}, where strategies are not necessarily cut-free.

\subsection{Simulation of visible pointer structures}

We introduce now the main tool of this paper, a reduction on agents which ``simulates" visible pointer structures: if $n[\{d_1\}a_1, \dots, \{d_p\}a_p]$ and $b$ are agents ($n>0$),
we define the non-deterministic reduction relation $\leadsto$ on triples $(a, d, b)$, where $d$ is a depth (a natural number) and $a$ and $b$ are agents, by the following two cases:
\begin{eqnarray*}
(n[\{d_1\}a_1, \dots, \{d_p\}a_p],d, b) &\leadsto& (a_i,d_i-1,(n-1)[\{d_1\}a_1, \dots, \{d_p\}a_p, \{d\}b])\\
(n[\{d_1\}a_1, \dots, \{d_p\}a_p],d, b) &\leadsto& (b, d-1, (n-1)[\{d_1\}a_1, \dots, \{d_p\}a_p, \{d\}b])
\end{eqnarray*}
where $i\in \{1, \dots , p\}$, $d_i>0$ in the first case and $d>0$ in the second case. We can now state the following central proposition.

\begin{proposition}[Simulation]
Let $(s,i)\in a\star_d b$, then if $s_{i+1}$ is defined, there exists $(a,d,b)\leadsto (a',d',b')$ such that $(s, i+1)\in a'\star_{d'} b'$.
\end{proposition}
\begin{proof}
The proof proceeds by a close analysis of where in its $P$-view (resp. $O$-view) $s_{i+1}$ can point. If it points to $s_i$, then the active strategy
asks for its argument which corresponds to the second reduction case. If it points to some element $s_{n_i}$ of its context, the active strategy calls the $i$-th
element of its context: this is the first reduction case, putting the subtree $a_i$ in head position. The rest of the proof consists in technical verifications,
to check that the new triple $(a', d', b')$ is such that $(s, i+1)\in a'\star_{d'} b'$.
\end{proof}

The result above will be sufficient for our purpose. Let us mention in passing that the connection between visible pointer structures
and agents is in fact tighter: a reduction chain starting from a triple $(n[], d, p[])$ can also be canonically mapped to a pointed visible pointer structure in $n\star_d p$,
and the two translations are inverse of one another. The interested reader is directed to \cite{phd}.

Before going on to the study of the rewriting rules introduced above, let us give a last simplification. If $a = n[\{d_1\}t_1, \dots, \{d_p\}t_q]$ and $b$ are agents,
then $a\cdot_d b$ will denote the agent obtained by appending $b$ as a new son of the root of $a$ with label $d$, \emph{i.e.} $n[\{d_1\}t_1, \dots, \{d_p\}t_q, \{d\}b]$.
Consider the following non-deterministic rewriting rule on agents:
\[
n[\{d_1\}a_1, \dots, \{d_p\}a_p] \leadsto a_i \cdot_{d_i - 1} (n-1)[\{d_1\}a_1, \dots, \{d_p\}a_p]
\]
Both rewriting rules on triples $(a, d, b)$ are actually instances of this reduction, by the isomorphism $(a, d, b) \mapsto a \cdot_d b$. We let the obvious verification
to the reader. This is helpful, as all that remains to study is this reduction on agents illustrated in Figure \ref{rewrite}. To summarize, if $N(a)$ denotes the length of the longest reduction sequence
starting from an agent $a$, we have the following property.

\begin{proposition}
Let $n, p\geq 0$, $d\geq 2$, then $N_d(n,p) \leq N(n[\{d\}p[]])+1$.
\label{equiv}
\end{proposition}
\begin{proof}
Obvious from the simulation lemma, adding $1$ for the initial move which is not accounted for by the reduction on agents. In fact this is an equality, as one can prove using the \emph{reverse} simulation
lemma mentioned above. See \cite{phd}.
\end{proof}

\begin{figure}
\[
\xymatrix{
&n\ar@{-}[ddr]^{d_p}
  \ar@{-}[ddl]_{d_1}&&&&a_i\ar@{-}[dl]
                    \ar@{-}[dr]
                \ar@{-}[drr]^{d_i-1}\\
                &&\ar[r]&&\ar@{.}[rr]&&&n-1
                                        \ar@{-}[dr]^{d_p}
                                        \ar@{-}[dl]_{d_1}\\
a_1\ar@{.}[rr]&&a_q&&&&a_1\ar@{.}[rr]&&a_q
}
\]
\caption{Rewriting rule on agents}
\label{rewrite}
\end{figure}
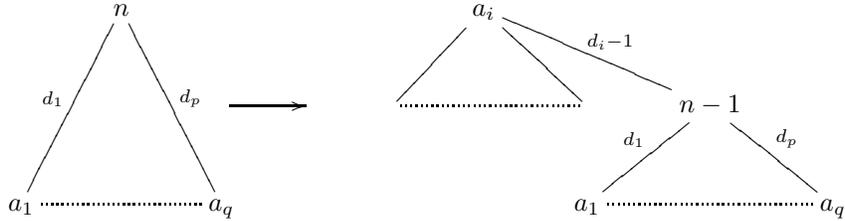

\section{Length of interactions}

The goal of this section is to study the reduction on agents introduced above, and to estimate its maximal length. We will first provide an upper bound for this length, adapting a 
method used by Beckmann \cite{beckmann2001exact} to estimate the maximal length of reductions on simply typed $\lambda$-calculus. We will then discuss the question of lower bounds, and finally
describe an application to head linear reduction.

\subsection{Upper bound}

We define on agents a predicate $\sststile{\rho}{\alpha}$, which introduction rules are compatible both with syntax and reduction.

\begin{definition}
The predicate $\sststile{\rho}{\alpha}$ (where $\rho, \alpha$ range over natural numbers) is defined on agents in the following inductive way.
\begin{itemize}
\item \textsc{Base.} $\sststile{\rho}{\alpha} 0[\{d_1\} a_1, \dots, \{d_p\} a_p]$
\item \textsc{Red.} Suppose $a = n[\{d_1\}a_1, \dots, \{d_p\} a_p]$. Then if for all $a'$ such that $a\leadsto a'$ we have $\sststile{\rho}{\alpha} a'$ and if
we also have $\sststile{\rho}{\alpha}(n-1)[\{d_1\}a_1, \dots, \{d_p\} a_p]$, then $\sststile{\rho}{\alpha+1} a$.
\item \textsc{Cut.} If $\sststile{\rho}{\alpha} a$, $\sststile{\rho}{\beta} b$ and $d\leq \rho$, then $\sststile{\rho}{\alpha+\beta} a \cdot_d b$.
\end{itemize}
\end{definition}

By this inductive definition, each proposition $\sststile{\rho}{\alpha} a$ is witnessed by a tree using \textsc{Base}, \textsc{Red} and \textsc{Cut}. \textsc{Red}-free trees look like
syntax trees, are easy to build but give few information on the reduction, whereas \textsc{Cut}-free trees look like reduction trees, are difficult to build but give very accurate information
on the length of reduction. The idea of the proof is then to design an automatic way to turn a \textsc{Red}-free tree to a \textsc{Cut}-free tree, \emph{via} a cut elimination lemma.
Let us now give the statement and sketch the proof of the four important lemmas that underlie our reasoning.

A \textbf{context-agent} $a()$ is a finite tree whose edges are labelled by natural numbers, and whose nodes are labelled either by natural numbers, or by the variable $x$, with the constraint that
all edges leading to $x$ must be labelled by the same number $d$; $d$ is called the \textbf{type} of $x$ in $a()$. If
We denote by $a(b)$ the result of substituting of all occurrences of $x$ in $a()$ by $b$. We denote by $a(\emptyset)$ the agent obtained by deleting in $a$
all occurrences of $x$, along with the edges leading to them.

\begin{lemma}[Substitution lemma]
If $\sststile{\rho}{\alpha} a(\emptyset)$, $\sststile{\rho}{\beta} b$ and $d \leq \rho + 1$ (where $d$ is the type of $x$ in $a$), then $\sststile{\rho}{\alpha(\beta+1)} a(b)$
\label{main_substitution}
\end{lemma}
\begin{proof}
We prove by induction on the tree witness for $\sststile{\rho}{\alpha} a(\emptyset)$ that the above property is true for all context-arena $a'()$ such that $a(\emptyset) = a'(\emptyset)$. The way
to handle each case is essentially forced by the induction hypothesis.
\end{proof}

\begin{lemma}[Cut elimination lemma]
Suppose $\sststile{\rho+1}{\alpha} a$. Then if $\alpha = 0$, $\sststile{\rho}{0} a$. Otherwise, $\sststile{\rho}{2^{\alpha-1}} a$.
\end{lemma}
\begin{proof}
By induction on the witness for $\sststile{\rho+1}{\alpha} a$, using the substitution lemma when the last rule is \textsc{Cut} with a type of $\rho+1$.
\end{proof}

\begin{lemma}[Recomposition lemma]
Let $a$ be an agent. Then $\sststile{depth(a)}{max(a)|a|} a$, where $depth(a)$ is the maximal label of an edge in $a$, $max(a)$ is the maximal label of a node and $|a|$ is the number of nodes.
\end{lemma}
\begin{proof}
By induction on $a$.
\end{proof}

\begin{lemma}[Bound lemma]
Let $a$ be an agent, then if $\sststile{0}{\alpha} a$, $N(a) \leq \alpha$.
\end{lemma}
\begin{proof}
The only used rules are \textsc{Base}, \textsc{Red} and \textsc{Cut} with $\rho = 0$. These \textsc{Cut} rules do not add any possible reduction and are easy to eliminate, then the lemma
is easily proved by induction on $a$.
\end{proof}

These lemmas are sufficient to give a first upper bound, by iterating the cut elimination lemma starting from the witness tree for $\sststile{\rho}{\alpha} a$ generated by the recomposition lemma. However when the
type is small, some of the lemmas above can be improved. For instance if $\sststile{0}{\alpha} a(\emptyset)$, $\sststile{0}{\beta} b$ and the type of $x$ in $a()$ is $1$, then $\sststile{0}{\alpha+\beta} a(b)$, 
since once the reduction reaches $b$ it will never enter $a()$ again. Using this we get a ``base" cut-elimination lemma, stating that for all $a$, whenever $\sststile{1}{\alpha} a$ then we have actually 
$\sststile{0}{\alpha} a$ instead of $\sststile{0}{2^{\alpha-1}} a$. Using this, we prove the following.

\begin{theorem}[Upper bound]
Let $depth(a)$ denote the highest edge label in $a$, $max(a)$ means the highest node label and $|a|$ means the number of nodes of $a$. Then if $depth(a)\geq 1$ and $max(a) \geq 1$ we have:
\[
N(a) \leq 2_{depth(a)-1}^{max(a)|a|-1}
\]
For the particular case when $a=n[\{d\}p[]]$ and if $d\geq 2$ we have:
\[
N_d(n, p) \leq 2_{d-2}^{n(p+1)}
\]
\end{theorem}
\begin{proof}
Both proofs are rather direct. For the first part, by the recomposition lemma we have $\sststile{depth(a)}{max(a)|a|} a$. It suffices then to apply $depth(a)-1$ times
the cut elimination lemma, then use the ``base" cut-elimination lemma to eliminate the remaining cuts. For the second part we reason likewise, but rely on the substitution lemma instead of the recomposition
lemma to get $\sststile{d-1}{n(p+1)} n[\{d\}p[]]$, which gives $N(n[\{d\}p[]]) \leq 2_{d-2}^{n(p+1)-1}$. But we have $N_d(n, p) \leq N(n[\{d\}p[]])+1$ by Proposition \ref{equiv}, which concludes the proof.
\end{proof}

Note that whereas the bounds in \cite{beckmann2001exact} are asymptotic and give poor quantitative information if instantiated on small types, our bound does provide valuable information on interactions with
small depth. For instance, if $\sigma : A\tto B$ and $\tau: B \tto C$ such that $\size(\sigma) = p$,
$\size(\tau) = n$ and the depth of $B$ is at most $2$, then no interaction between $\sigma$ and $\tau$ can be longer than $N_3(n, p) \leq 2^{n(p+1)}$. As we will see below, this can not be significantly improved.
In fact, we conjecture that for all $n\geq 1$ and $p\geq 2$, we have $N_3(n,p) = 2\frac{p^n-1}{p- 1}+1$ : this was found and machine-checked for all $n+p \leq 17$ thanks to an implementation of agents and their
reduction, unfortunately we could not prove its correctness, nor generalize it to higher depths.

\subsection{Lower bound}

As argued in the introduction, the upper bound above applies to several programming languages executed by head linear reduction, possibly featuring non determinism and/or ground type references, therefore the fact that
we used game semantics to prove it increases its generality. On the other hand, if we try to give the closest possible lower bound for $N_d(n,p)$ using the full power of visible pointer structures, we would get a lower
bound without meaning for most languages concerned by the upper bound, since pointer structures have no innocence or determinism requirements\footnote{Our experiments with pointer structures and agents confirmed indeed that the
possibility to use non-innocent behaviour does allow significantly longer plays.}. Therefore what makes more sense is to describe a lower bound in the more restricted possible framework, \emph{i.e.}
simply typed $\lambda$-calculus.

We won't detail the construction much, as the method is standard and does not bring a lot to our analysis. The idea is to define higher types for church integers by $A_0 = \bot$ and $A_{n+1} = A_n \to A_n$. Then, 
denoting by $\church{n}_p$ the church integer for $n$ of type $A_{p+2}$, we define $S_n = \church{2}_n \church{2}_{n-1} \dots \church{2}_0 : A_2$. We apply then $S_n$ to $id_\bot$ to get a term whose head linear
reduction chain has at least $2_{n+1}^1$ steps. In game semantics, $\intr{\church{2}_n}$ has size $n+3$ and all other components have size smaller than $n+2$, the depth of the ambient arena being $n+2$. The function $N_d(n,p)$ being
monotonically increasing in all its parameters we have the following inequalities for $3 \leq d \leq min(n-1, p)$, both bounds making sense for all programming languages containing the simply-typed $\lambda$-calculus and whose 
terms can be interpreted as bounded strategies.
\[
2_{d-2}^2 \leq N_d(n,p) \leq 2_{d-2}^{n(p+1)}
\]
Note that from this we can deduce bounds for $N(n,p)$, when we have no information on the depth of the ambient arena. Indeed, we always have $d\leq 2n$ and $d\leq 2p+1$ because a pointer chain in a play
is visible by both players. Thus, $N(n,p) = N_{min(2n, 2p+1)}(n,p)$.

\subsection{Application to head linear reduction}

Earlier works on game semantics \cite{danosregnier} suggest that in every games model of a programming language lies a hidden notion of linear reduction, head linear reduction
when modelling call-by-name evaluation: this is the foundation for our claim that our game-theoretic result is really about the length of execution in programming languages whose terms can
be described as bounded strategies. Of course it requires some work to interface execution in these programming languages to our game-theoretic results, and part of this work has to be redone in each case. 
To illustrate this, we now describe how to extract from our results a theorem about the length of head linear reduction sequences in simply-typed $\lambda$-calculus. For the formal definition of head linear reduction, the reader
is directed to \cite{danos:abstract}. If $S$ is a $\lambda$-term then the \textbf{spinal height} of $S$ is the quantity $sh(S)$ defined by induction as $sh(x) = 1$, $sh(\lambda x.S) = sh(S)$ and $sh(S T) = max(sh(S), sh(T)+1)$; when $S$ is a $\beta\eta$-normal form, $sh(S)$ is nothing but the height of its Böhm tree.
The \textbf{height} of $S$ is the subtly different quantity\footnote{One can easily prove that on closed terms, it is always less than the more common notion of height defined as
$h(x) = 0$, $h(\lambda x. S) = 1 + h(S)$ and $h(S T) = max(h(S), h(T))+1$, for which our upper bound consequently also holds.} $h(S)$ defined by $h(x) = 1$, $h(\lambda x. M) = h(M)$ and $h(M N) = max(h(M), h(N)) + 1$. Finally, the
\textbf{level} of a type $lv(A)$ is defined by $lv(\bot) = 0$ and $lv(A\to B) = max(lv(A) + 1, lv(B))$ and the \textbf{degree} $g(S)$ of a term is the maximal level of the type of all subterms
of $S$.

A \textbf{game situation} \cite{phd} is the data of $\lambda$-terms $S:A_1 \to \dots \to A_p \to B$ and $T_1: A_1, \dots T_p:A_p$ in $\eta$-long $\beta$-normal form, and we are interested
in the term $S T_1 \dots T_p$.  Our game-theoretic results
apply immediately to game situations, because of the connection between game-theoretic interaction and head linear reduction \cite{danosregnier}:
if $N(S T_1 \dots T_p)$ denotes the length of the head linear reduction chain of $S T_1 \dots T_p$, then we have 
$N(S T_1 \dots T_p) \leq N_d(n, p)$ where
$d$ is the depth of the arena corresponding to $A\to B$, $n$ is the size of $\intr{S}$ and $p$ is the maximal size of all of the $\intr{T_i}$. But since $S$ and $T_i$ are already in $\eta$-long
$\beta$-normal form, we have $|\intr{S}| = sh(S)$ and $|\intr{T_i}| = sh(T_i)$. Thus, we conclude that in the case of a game situation we have:
\[
N(S T_1 \dots T_p) \leq 2_{max_i lv(A_i)-1}^{sh(S)(max_i(sh(T_i))+1)}
\]
Outside of game situations, it is less obvious to see how our results apply. The more elegant approach would be probably to extend the connection between head linear reduction
and game semantics to \emph{revealed} game semantics, which would give the adequate theoretical foundations to associate an agent to any $\eta$-long $\lambda$-term. Without
these tools, we can nonetheless apply the following hack. Suppose we have a $\lambda$-term $S$. The idea is to ``delay" all redexes, replacing each redex $(\lambda x. S) T$ of type $A\to B$ in $S$ with 
$y_{A, B} (\lambda x.S) T$, where we add a new symbol $y_{A, B}: (A\to B) \to A \to B$ for each pair $(A, B)$. We iterate this operation until we reach a $\beta$-normal $\lambda$-term $S^t$, which satisfies $sh(S^t) \leq h(S)$.
We then expand $S^t$ to its $\eta$-long form $\eta(S^t)$, which satisfies $sh(\eta(S^t)) \leq sh(S^t)+g(S^t) \leq h(S) + g(S) + 1$. We consider now the term $(\lambda y_1 \dots y_p. \eta(S^t)) ev_1 \dots ev_p$, 
where each $y_i$ binds one of the new symbols $y_{A, B}$, and $ev_i:(A \to B) \to A \to B$ is the ($\eta$-long form of) the corresponding evaluation $\lambda$-term. We recognise here a game situation, whose head linear
reduction chain is necessarily longer than for $S$ (we have only added steps due to the delaying of redexes and $\eta$-expansion). Using the inequality above for game situations, we conclude:
\[
N(S) \leq 2_{g(S)}^{(h(S) + g(S) + 1)(g(S) + 1)}
\]

\section{Conclusion \& future work}

Applied to head linear reduction on simply typed $\lambda$-calculus, our results show that the price of linearity is not as high as one might expect. 
Not only the bounds remain in $\mathcal{E}^4$, but they are only slightly higher than those for usual $\beta$-reduction: in particular, the height 
of the tower of exponentials is the same.

A strength of our method is that it is not restricted to $\lambda$-calculus; the results should indeed immediately apply as well to similar notions of reduction on other total
programming languages. Beyond ground type references and non determinism, there are also games model of call-by-value languages \cite{abramsky-mccusker:families} generating pointer structures as well, thus this work should also provide
bounds for the corresponding call-by-value linear reduction (\emph{tail} linear reduction?). All the tools used here also can be extended to \emph{non-alternating plays} \cite{DBLP:conf/concur/Laird05}, which suggests that this work could
be used to give bounds to the length of reductions in some restricted concurrent languages.

We also believe \emph{agents} are worth studying further. Their combinatorial nature and their connection to execution of programs may prove interesting for the study of higher order systems with restricted complexity, 
such as \emph{light} linear logics \cite{DBLP:journals/iandc/Girard98}. For instance, proofs typable in light systems may correspond to agents with some restricted behaviours, which would make them a valuable tool for 
the study of programming languages with implicit complexity.

\paragraph{Acknowledgements.} This work was partially supported by the French ANR project CHOCO. The author also would like to thank Fabien Renaud for interesting discussions on related subjects.


\input{appendix.tex}

\end{document}

%% file: appendix.tex
\newpage

\appendix

\section{pointer structures and rewriting}

\begin{lemma}
Let $s$ be a pointed visible pointer structure and $a=n[\{d_1\}a_1, \dots, \{d_p\}a_p]$ an agent such that $(s, i)\in coTr(a)$. Then
if $s_j\rightarrow s_i$, $(s, j)\in Tr(a)$.
\label{lem_point}
\end{lemma}
\begin{proof}
Let us suppose without loss of generality that $s_i$ is an Opponent move; the other case can be obtained just by switching Player/Opponent and $P$-views/$O$-views
everywhere. Then $s_j$ being a Player move, we have to check first that $\size(s, j)\leq 2n$, \emph{i.e.}
\[
\max_{s_j \in \oview{s_{\leq k}}} |\oview{s_{\leq k}}| - |\oview{s_{\leq j}}| + 1 \leq 2n
\]
We use that $\cosize(s, i) \leq 2n+1$, \emph{i.e.}
\[
\max_{s_i \in \oview{s_{\leq k}}} |\oview{s_{\leq k}}| - |\oview{s_{\leq i}}| + 1 \leq 2n+1
\]
But $s_j \rightarrow s_i$, hence $|\oview{s_{\leq j}}| = |\oview{s_{\leq i}}| + 1$ and the inequality is obvious. We need now to examine
the context of $(s, j)$. Since $s_j$ is a Player move, it is defined as the set $\{s_{n_1}, \dots, s_{n_p}\}$ of Player moves appearing in
$\oview{s_{<j}}$, which is also the set of Player moves appearing in $\oview{s_{<i}}$ and therefore the co-context of $(s, i)$. But $(s, i)\in coTr(a)$,
hence for all $k\in \{1, \dots, p\}$ we have $(s, n_k)\in coTr(a_k)$ which is exactly what we needed.
\end{proof}

\begin{proposition}[Simulation]
Let $(s,i)\in a\star_d b$, then if $s_{i+1}$ is defined, there exists $(a,d,b)\leadsto (a',d',b')$ such that $(s, i+1)\in a'\star_{d'} b'$.
\end{proposition}
\begin{proof}
Suppose $a = n[\{d_1\}a_1, \dots, \{d_p\}a_p]$.
Let $\{s_{n_1}, \dots, s_{n_p}\}$ be the context of $(s, i)$. By visibility, $s_{i+1}$ must either point to $s_i$ or to an element of the context. Let us
distinguish cases.
\begin{itemize}
\item If $s_{i+1} \rightarrow s_i$, then we claim that $(s, i+1) \in b\star_{d-1} (n-1)[\{d_1\}a_1, \dots, \{d_p\}a_p, \{d\}b]$, \emph{i.e} $(s, i+1) \in Tr(b)$,
$(s, i+1) \in coTr((n-1)[\{d_1\}a_1, \dots, \{d_p\}a_p, \{d\}b])$ and the depth of $s$ relative to $i+1$ is at most $d-1$. For the first part, we use
that $(s, i)\in a\star_d b$ : in particular, $(s, i)\in coTr(b)$ and since $s_{i+1} \rightarrow s_i$ this implies by Lemma \ref{lem_point} that $(s, i+1)\in Tr(b)$.
For the second part, we must first check that $\cosize(s, i+1) \leq 2(n-1) + 1$. Let us suppose without loss of generality that $s_i$ is an Opponent move, all the
reasoning below can be adapted by switching Player/Opponent and $P$-views/$O$-views everywhere. We want to prove:
\[
\cosize(s,i+1)  = \max_{s_{i+1} \in \pview{s_{\leq j}}} |\pview{s_{\leq j}}| - |\pview{s_{\leq i+1}}| + 1 \leq 2(n-1) + 1
\]
But since $(s, i) \in Tr(a)$, we already know:
\[
\size(s, i) = \max_{s_i \in \pview{s_{\leq j}}} |\pview{s_{\leq j}}| - |\pview{s_{\leq i}}| + 1 \leq 2n
\]
Thus we only need to remark that $|\pview{s_{\leq i+1}}| = |\pview{s_{\leq i}}| + 1$ since $s_{i+1}$ is a Player move.
Now, we must examine the co-context of $(s, i+1)$, but by definition of $P$-view it is $\{s_{n_1}, \dots, s_{n_p}, s_i\}$ where $\{s_{n_1}, \dots, s_{n_p}\}$
is the context of $(s, i)$. Since $(s,i) \in Tr(n[a_1, \dots, a_p])$ we have as required $(s,n_k)\in coTr(a_k)$ for each $k\in \{1, \dots, p\}$
and $(s,i)\in coTr(b)$ because $(s, i)\in a\star_d b$. For the third part, we have to prove that the depth of $s$ relative to $i+1$ is at most $d-1$, but
it is obvious since the depth relative to $i$ is at most $d$ and $s_{i+1}\rightarrow s_i$.

\item Otherwise, we have $s_{i+1} \rightarrow s_{n_j}$ for $j\in \{1, \dots, p\}$. Then, we claim
that $(s, i+1) \in a_j \star_{d_i - 1} (n-1)[\{d_1\}a_1, \dots, \{d_p\}a_p, \{d\}b]$. We do have
$(s, i+1) \in Tr(a_j)$ because $(s, i)\in Tr(n[\{d_1\}a_1, \dots, \{d_p\}a_p])$, thus $(s, n_j)\in coTr(a_j)$ and $(s, i+1) \in Tr(a_j)$ by Lemma \ref{lem_point}.
It remains to show that $(s, i+1) \in coTr((n-1)[\{d_1\}a_1, \dots, \{d_p\}a_p, \{d\}b])$ and that the depth of $s$ relative to $i+1$ is at most $d_1 - 1$, but
the proofs are exactly the same as in the previous case.
\end{itemize}
\end{proof}

\section{Upper bound}

\begin{lemma}[Monotonicity]
If $\sststile{\rho}{\alpha} a$, then $\sststile{\rho'}{\alpha'} a$ for all $\alpha \leq \alpha'$ and $\rho \leq \rho'$.
\label{monotonicity}
\end{lemma}
\begin{proof}
By induction on $a$.
\end{proof}

\begin{lemma}[Null substitution lemma]
If $\sststile{\rho}{\alpha} a(\emptyset)$ and the type of $x$ in $a$ is $0$, then for all agent $b$ we still have $\sststile{\rho}{\alpha} a(b)$. Moreover, the witness includes as many \textsc{Cut} rules
as for $\sststile{\rho}{\alpha} a(\emptyset)$.
\label{first_substitution}
\end{lemma}
\begin{proof}
We prove by induction on the tree witness for $\sststile{\rho}{\alpha} a(\emptyset)$ that the above property is true for all context-arena $a'$ such that $a(\emptyset) = a'(\emptyset)$.
\begin{itemize}
\item \textsc{Base.} The root of $a$ is $0$, hence the result is trivial.
\item \textsc{Red.} Suppose $a'$ has the form $n[\{d_1\}a_1, \dots, \{d_p\}a_p, \{d\}x]$, where $a_1, \dots, a_p$ possibly include occurrences of $x$ (the case where $x$ appears as a son
of the root encompasses the other). The premises of \textsc{Red} are then that for $1 \leq i \leq p$ such that
$d_i\geq 1$, $\sststile{\rho}{\alpha-1} a_i(\emptyset) \cdot_{d_i-1} (n-1)[\{d_1\}a_1(\emptyset), \dots \{d_p\}a_p(\emptyset)]$
and $\sststile{\rho}{\alpha-1} (n-1)[\{d_1\}a_1(\emptyset), \dots \{d_p\}a_p(\emptyset)]$. The induction hypothesis on these premises give witnesses for the two following properties:
\begin{eqnarray}
\sststile{\rho}{\alpha-1} (a_i \cdot_{d_i-1} (n-1)[\{d_1\}a_1, \dots, \{d_p\}a_p, \{d\}x])(b)\label{eq1}
\end{eqnarray}
\begin{eqnarray}
\sststile{\rho}{\alpha-1} ((n-1)[\{d_1\}a_1, \dots \{d_p\}a_p, \{d\}x])(b)\label{eq2}
\end{eqnarray}
All the possible reductions are already covered since $d=0$, thus by \textsc{Red} we have $\sststile{\rho}{(\alpha-1)+1} a(b)$ as required.
\item \textsc{Cut.} Let us suppose $\sststile{\rho}{\alpha+\gamma} a(\emptyset)$ is obtained by \textsc{Cut}, hence $a(\emptyset)$ has the form $a_1(\emptyset) \cdot_{d_1} a_2(\emptyset)$. Let us suppose
that $a'$ has the form $(a_1 \cdot_{d'} a_2) \cdot_d x$, since once again the case where $x$ is a child of the root of $a'$ encompasses the other. The premises of \textsc{Cut} are then
$\sststile{\rho}{\alpha} a_1(\emptyset)$ and $\sststile{\rho}{\gamma} a_2(\emptyset)$, and $d'\leq \rho$. Note now that we also have $(a_1 \cdot_d x)(\emptyset) = a_1(\emptyset)$, therefore
the induction hypothesis on $\sststile{\rho}{\alpha} a_1(\emptyset)$ along with $\sststile{\rho}{\beta} b$ and $d\leq \rho+1$ implies that
$\sststile{\rho}{\alpha} a_1(b) \cdot_d b$. But by induction hypothesis we also have $\sststile{\rho}{\gamma}a_2(b)$, hence by \textsc{Cut}:
\[
\sststile{\rho}{\alpha + \gamma} (a_1(b) \cdot_d b) \cdot_{d'} a_2(b)
\]
Which was what was required for $(a_1(b) \cdot_{d'} a_2(b)) \cdot_d b$, thus it suffices since trees are considered up to permutation.
\end{itemize}
\end{proof}

\begin{lemma}[Main substitution lemma]
If $\sststile{\rho}{\alpha} a(\emptyset)$, $\sststile{\rho}{\beta} b$ and $d \leq \rho + 1$ (where $d$ is the type of $x$ in $a$), then $\sststile{\rho}{\alpha(\beta+1)} a(b)$
\label{main_substitution}
\end{lemma}
\begin{proof}
We prove by induction on the tree witness for $\sststile{\rho}{\alpha} a(\emptyset)$ that the above property is true for all context-arena $a'$ such that $a(\emptyset) = a'(\emptyset)$.
\begin{itemize}
\item \textsc{Base.} The root of $a$ is $0$, hence the result is trivial.
\item \textsc{Red.} Suppose $a'$ has the form $n[\{d_1\}a_1, \dots, \{d_p\}a_p, \{d\}x]$, where $a_1, \dots, a_p$ possibly include occurrences of $x$ (the case where $x$ appears as a son
of the root encompasses the other). The premises of \textsc{Red} are then that for $1 \leq i \leq p$ such that
$d_i\geq 1$, $\sststile{\rho}{\alpha-1} a_i(\emptyset) \cdot_{d_i-1} (n-1)[\{d_1\}a_1(\emptyset), \dots \{d_p\}a_p(\emptyset)]$
and $\sststile{\rho}{\alpha-1} (n-1)[\{d_1\}a_1(\emptyset), \dots \{d_p\}a_p(\emptyset)]$. The induction hypothesis on these premises give witnesses for the two following properties:
\begin{eqnarray}
\sststile{\rho}{(\alpha-1)(\beta+1)} (a_i \cdot_{d_i-1} (n-1)[\{d_1\}a_1, \dots, \{d_p\}a_p, \{d\}x])(b)\label{eq1}
\end{eqnarray}
\begin{eqnarray}
\sststile{\rho}{(\alpha-1)(\beta+1)} ((n-1)[\{d_1\}a_1, \dots \{d_p\}a_p, \{d\}x])(b)\label{eq2}
\end{eqnarray}
By hypothesis we have $\sststile{\rho}{\beta} b$, hence by \textsc{Cut} (since $d-1\leq \rho$), we have:
\begin{eqnarray}
\sststile{\rho}{(\alpha-1)(\beta+1) + \beta} b \cdot_{d-1} (n-1)[\{d_1\}a_1(b), \dots \{d_p\}a_p(b), \{d\}b]\label{eq3}
\end{eqnarray}
Using (\ref{eq1}) for all $i\in \{1, \dots, p\}$, (\ref{eq2}) (adjusted to $\sststile{\rho}{(\alpha-1)(\beta+1)+\beta}$ by Lemma \ref{monotonicity}) and (\ref{eq3}) we deduce by \textsc{Red} that
\[
\sststile{\rho}{(\alpha-1)(\beta+1)+ \beta + 1} n[\{d_1\}a_1(b), \dots, \{d_p\}a_p(b), \{d\}b]
\]
Which is what was required.
\item \textsc{Cut.} Let us suppose $\sststile{\rho}{\alpha+\gamma} a(\emptyset)$ is obtained by \textsc{Cut}, hence $a(\emptyset)$ has the form $a_1(\emptyset) \cdot_{d_1} a_2(\emptyset)$. Let us suppose
that $a'$ has the form $(a_1 \cdot_{d'} a_2) \cdot_d x$, since once again the case where $x$ is a child of the root of $a'$ encompasses the other. The premises of \textsc{Cut} are then
$\sststile{\rho}{\alpha} a_1(\emptyset)$ and $\sststile{\rho}{\gamma} a_2(\emptyset)$, and $d'\leq \rho$. Note now that we also have $(a_1 \cdot_d x)(\emptyset) = a_1(\emptyset)$, therefore
the induction hypothesis on $\sststile{\rho}{\alpha} a_1(\emptyset)$ along with $\sststile{\rho}{\beta} b$ and $d\leq \rho+1$ implies that
$\sststile{\rho}{\alpha(\beta+1)} a_1(b) \cdot_d b$. But by induction hypothesis we also have $\sststile{\rho}{\gamma(\beta+1) }a_2(b)$, hence by \textsc{Cut}:
\[
\sststile{\rho}{\alpha(\beta+1) + \gamma(\beta+1)} (a_1(b) \cdot_d b) \cdot_{d'} a_2(b)
\]
Which was what was required for $(a_1(b) \cdot_{d'} a_2(b)) \cdot_d b$, thus it suffices since trees are considered up to permutation.
\end{itemize}
\end{proof}

\begin{lemma}[Cut elimination lemma]
Suppose $\sststile{\rho+1}{\alpha} a$. Then if $\alpha = 0$, $\sststile{\rho}{0} a$. Otherwise, $\sststile{\rho}{2^{\alpha-1}} a$.
\end{lemma}
\begin{proof}
By induction on the tree witness for $\sststile{\rho+1}{\alpha} a$.
\begin{itemize}
\item \textsc{Base.} Trivial.
\item \textsc{Red.} Suppose $a=n[\{d_1\}a_1, \dots, \{d_p\}a_p]$, the premises of \textsc{Red} are
$\sststile{\rho + 1}{\alpha-1} a_i \cdot_{d_i - 1} (n-1)[\{d_1\}a_1, \dots, \{d_p\}a_p]$ for all $i\in \{1, \dots, p\}$ and $\sststile{\rho+1}{\alpha-1} (n-1)[\{d_1\}a_1, \dots, \{d_p\}a_p]$.
If $\alpha \geq 2$, then it follows by induction hypothesis that $\sststile{\rho}{2^{\alpha-2}} a_i \cdot_{d_i - 1} (n-1)[\{d_1\}a_1, \dots, \{d_p\}a_p]$ and
$\sststile{\rho}{2^{\alpha-2}}(n-1)[\{d_1\}a_1, \dots, \{d_p\}a_p]$, which implies by \textsc{Red} and Lemma \ref{monotonicity} that $\sststile{\rho}{2^{\alpha-1}} a$. If $\alpha = 1$, then
the premises of \textsc{Red} are $\sststile{\rho + 1}{0} a_i \cdot_{d_i - 1} (n-1)[\{d_1\}a_1, \dots, \{d_p\}a_p]$ for all $i\in \{1, \dots, p\}$ and $\sststile{\rho+1}{0} (n-1)[\{d_1\}a_1, \dots, \{d_p\}a_p]$. By induction
hypothesis this is still true with $\rho$ instead of $\rho+1$, thus by \textsc{Red} we have $\sststile{\rho}{1} [\{d_1\}a_1, \dots, \{d_p\}a_p]$ which is what we needed to prove.
\item \textsc{Cut.} Suppose $a= a_1 \cdot_d a_2$, the premises of \textsc{Cut} are $\sststile{\rho+1}{\alpha} a_1$, $\sststile{\rho+1}{\beta} a_2$ and $d\leq \rho+1$. If $\alpha, \beta \geq 1$
then by induction hypothesis it follows that $\sststile{\rho}{2^{\alpha-1}} a_1$ and $\sststile{\rho}{2^{\beta-1}} a_2$, in particular if we define a context-agent $a'_1 = a_1 \cdot_d x$ we have
$\sststile{\rho}{2^{\alpha-1}} a'_1(\emptyset)$, hence by the substitution lemma (since $d\leq \rho+1$) we have
$\sststile{\rho}{2^{\alpha-1} (2^{\beta-1}+1)} a'_1(a_2) = a_1 \cdot_d a_2 = a$, thus $\sststile{\rho}{2^{\alpha+\beta-1}} a$ thanks to Lemma \ref{monotonicity} (since it is always true
than $2^{\alpha+\beta-1} \geq 2^{\alpha-1} (2^{\beta-1}+1)$).
If $\alpha = 0$ then by induction hypothesis we have $\sststile{\rho}{0} a_1$ and $\sststile{\rho}{\beta'} a_2$. We use then the substitution lemma (since $d\leq \rho+1$) to get $\sststile{\rho}{0}(a_1 \cdot_d a_2)$, which
is stronger that what was required whatever was the value of $\beta$. The last remaining case is when $\alpha=1$ and $\beta=0$, then by induction hypothesis $\sststile{\rho}{1} a_1$ and $\sststile{\rho}{0} a_1$,
thus by the substitution lemma we have as required $\sststile{\rho}{1}(a_1 \cdot_d a_2)$.
\end{itemize}
\end{proof}

\begin{lemma}[Recomposition lemma]
Let $a$ be an agent. Then:
\[
\sststile{depth(a)}{max(a)|a|} a
\]
Where $depth(a)$ is the maximal label of an edge in $a$, $max(a)$ is the maximal label of a node and $|a|$ is the number of nodes.
\end{lemma}
\begin{proof}
First, let us show that the following rule \textsc{Base'} is admissible, for any $\alpha$ and $\rho$.
\[
\sststile{\rho}{\alpha + n} n[]
\]
If $n=0$ this is exactly \textsc{Base}. Otherwise we apply \textsc{Red}. There is no possible reduction, so the only thing we have to prove is $\sststile{\rho}{\alpha + n - 1} (n-1)[]$, which
is provided by the induction hypothesis. Then we prove the lemma by immediate induction on $a$, using only \textsc{Base'}, \textsc{Cut} and Lemma \ref{monotonicity}.
\end{proof}

From now on, let $N(a)$ denote the longest reduction sequence of $a$. We also use the notations $2_0^n = n$ and $2_{d+1}^n = 2^{{2_d^n}}$ for iterated exponentials.

\begin{lemma}[Bound lemma]
Let $a$ be an agent, then if $\sststile{0}{\alpha} a$, $N(a) \leq \alpha$.
\label{bound}
\end{lemma}
\begin{proof}
First of all we prove that if there is a witness for $\sststile{0}{\alpha} a$, then it can be supposed \textsc{Cut}-free: this is proved by induction on $\sststile{0}{\alpha} a$, eliminating
each use of \textsc{Cut} by Lemma \ref{first_substitution}. Then, by induction on the \textsc{Cut}-free witness tree for $\sststile{0}{\alpha} a$:
\begin{itemize}
\item \textsc{Base.} Then, the root of $a$ is $0$, thus $N(a) = 0$; there is nothing to prove.
\item \textsc{Red.} The premises of $\sststile{0}{\alpha} a$ include in particular that for all $a'$ such that $a\leadsto a'$, we have $\sststile{0}{\alpha-1} a'$. By induction hypothesis, this
means that for all such $a'$ we have $N(a') \leq \alpha-1$, hence $N(a) \leq \alpha$.
\end{itemize}
\end{proof}

From all this, it is possible to give a first upper bound by using the recomposition lemma, then iterating the cut elimination lemma. However, we will first prove here a refined version of the
cut elimination lemma when $\rho=1$, which will allow to decrease by one the height of the tower of exponentials. First, we need the following adaptation of the substitution lemma:

\begin{lemma}[Base substitution lemma]
If $\sststile{0}{\alpha} a(\emptyset)$, $\sststile{0}{\beta} b$ and the type of $x$ in $a$ is $1$, then $\sststile{0}{\alpha+\beta} a(b)$.
\label{specialized_substitution}
\end{lemma}
\begin{proof}
We prove by induction on the tree witness for $\sststile{0}{\alpha} a(\emptyset)$ that the above property is true for all context-arena $a'$ such that $a(\emptyset) = a'(\emptyset)$.
\begin{itemize}
\item \textsc{Base.} The root of $a$ is $0$, hence the result is trivial.
\item \textsc{Red.} Suppose $a'$ has the form $n[\{d_1\}a_1, \dots, \{d_p\}a_p, \{d\}x]$, where $a_1, \dots, a_p$ possibly include occurrences of $x$ (the case where $x$ appears as a son
of the root encompasses the other). The premises of \textsc{Red} are then that for $1 \leq i \leq p$ such that
$d_i\geq 1$, $\sststile{0}{\alpha-1} a_i(\emptyset) \cdot_{d_i-1} (n-1)[\{d_1\}a_1(\emptyset), \dots \{d_p\}a_p(\emptyset)]$
and $\sststile{0}{\alpha-1} (n-1)[\{d_1\}a_1(\emptyset), \dots \{d_p\}a_p(\emptyset)]$. The induction hypothesis on these premises give witnesses for the two following properties:
\begin{eqnarray}
\sststile{0}{\alpha-1+\beta} (a_i \cdot_{d_i-1} (n-1)[\{d_1\}a_1, \dots, \{d_p\}a_p, \{d\}x])(b)\label{eq4}
\end{eqnarray}
\begin{eqnarray}
\sststile{0}{\alpha-1 +\beta} ((n-1)[\{d_1\}a_1, \dots \{d_p\}a_p, \{d\}x])(b)\label{eq5}
\end{eqnarray}
By hypothesis we have $\sststile{0}{\beta} b$, hence by Lemma \ref{first_substitution} (since $d=1$) we have
\[
\sststile{0}{\beta} b \cdot_{d-1} (n-1)[\{d_1\}a_1(b), \dots \{d_p\}a_p(b), \{d\}b]\label{eq6}
\]
Hence, using (\ref{eq4}) for all $i\in \{1, \dots, p\}$, (\ref{eq5}) and (\ref{eq6}) (adjusted to $\sststile{0}{\alpha-1+\beta} b$ by Lemma \ref{monotonicity}) we deduce by \textsc{Red} that
\[
\sststile{0}{\alpha + \beta} n[\{d_1\}a_1(b), \dots, \{d_p\}a_p(b), \{d\}b]
\]
Which is what was required.
\item \textsc{Cut.} Let us suppose $\sststile{0}{\alpha+\gamma} a(\emptyset)$ is obtained by \textsc{Cut}, hence $a(\emptyset)$ has the form $a_1(\emptyset) \cdot_0 a_2(\emptyset)$. Let us suppose
that $a'$ has the form $(a_1 \cdot_0 a_2) \cdot_d x$. The premises of \textsc{Cut} are then
$\sststile{0}{\alpha} a_1(\emptyset)$ and $\sststile{0}{\gamma} a_2(\emptyset)$. Note now that we also have $(a_1 \cdot_d x)(\emptyset) = a_1(\emptyset)$, therefore
the induction hypothesis on $\sststile{\rho}{\alpha} a_1(\emptyset)$ along with $\sststile{0}{\beta} b$ implies that
$\sststile{0}{\alpha+\beta} a_1(b) \cdot_1 b$ and all that remains is to substitute $a_2(b)$ in $(a_1(b) \cdot_1 b) \cdot_0 x$. But since the type of $x$ is $0$, Lemma \ref{first_substitution} proves
that $\sststile{0}{\alpha+\beta} (a_1(b) \cdot_d b) \cdot_{d'} a_2(b)$, which concludes since trees are considered up to permutation.
\end{itemize}
\end{proof}

\begin{lemma}[Base cut elimination lemma]
If $\sststile{1}{\alpha} a$, then $\sststile{0}{\alpha} a$.
\label{specialized_cut_elim}
\end{lemma}
\begin{proof}
By induction on the witness tree for $\sststile{1}{\alpha} a$.
\begin{itemize}
\item \textsc{Base.} Trivial.
\item \textsc{Red.} Suppose $a$ has the form $n[\{d_1\}a_1, \dots, \{d_p\}a_p]$. The premises of \textsc{Red} are that for all $i\in \{1, \dots, p\}$ we have
$\sststile{1}{\alpha-1} a_1 \cdot_{d_i - 1} (n-1)\{d_1\}a_1, \dots, \{d_p\}a_p]$ and $\sststile{1}{\alpha-1} (n-1)\{d_1\}a_1, \dots, \{d_p\}a_p]$. The result is then
trivial by induction hypothesis and \textsc{Red}.
\item \textsc{Cut.} Suppose $a = a_1 \cdot_d a_2$ with $d\leq 1$, the premises of \textsc{Cut} are that $\sststile{1}{\alpha} a_1$ and $\sststile{1}{\beta} a_2$. If $d=0$, then the result
is trivial by the induction hypothesis and \textsc{Cut}. If $d=1$, we just apply Lemma \ref{specialized_substitution} instead of \textsc{Cut}.
\end{itemize}
\end{proof}

\begin{theorem}[Upper bound]
Let $depth(a)$ denote the highest edge label in $a$, $max(a)$ means the highest node label and $|a|$ means the number of nodes of $a$. Then if $depth(a)\geq 1$ and $max(a) \geq 1$ we have:
\[
N(a) \leq 2_{depth(a)-1}^{max(a)|a|-1}
\]
For the particular case when $a=n[\{d\}p[]]$ and if $d\geq 2$ we have:
\[
N_d(n, p) \leq 2_{d-2}^{n(p+1)}
\]
\end{theorem}
\begin{proof}
Let us first prove the first part. By the recomposition lemma, we have $\sststile{depth(a)}{max(a)|a|} a$. By $depth(a)-1$ iterations of the cut elimination lemma, we have
$\sststile{1}{2_{depth(a)-1}^{max(a)|a|-1}} a$. But then by Lemma \ref{specialized_cut_elim} we also have $\sststile{0}{2_{depth(a)-1}^{max(a)|a|-1}} a$. By Lemma \ref{bound}, this
implies as required that $N(a) \leq 2_{depth(a)-1}^{max(a)|a|-1}$. We turn now to the second part. Obviously, we have $\sststile{d-1}{n} n[]$ and $\sststile{d-1}{p} p[]$. By the
substitution lemma, this implies that $\sststile{d-1}{n(p+1)} n[\{d\}p[]]$. By $d-2$ applications of the cut elimination lemma, and one application of Lemma \ref{specialized_cut_elim}, this means
that $\sststile{0}{2_{d-1}^{n(p+1)-1}} [\{d\}p[]]$ hence $N_d(n,p) \leq 2_{d-2}^{n(p+1)-1}$.
\end{proof}